\documentclass[preprint,authoryear, 3p, 12pt]{elsarticle}

\usepackage{amssymb}

\journal{}

\newtheorem{theorem}{Theorem}[section]
\newtheorem{prop}[theorem]{Proposition}

\newtheorem{lemma}[theorem]{Lemma}

\newtheorem{define}[theorem]{Definition}
\newtheorem{example}[theorem]{Example}

\newenvironment{proof}{\noindent {\em Proof:\ }}{$\square$}

\def\gr{{Gr\"obner }}
\def\N{{\mathbb{N}}}
\def\Q{{\mathbb{Q}}}
\def\xi{{\frac{\partial}{\partial x_i}}}
\def\x1{{\frac{\partial}{\partial x_1}}}

\def\R{{\mathcal R}}
\def\A{{\mathcal A}}
\def\B{{\mathcal B}}
\def\J{{\mathcal J}}
\def\deg{{\rm deg}}
\def\lc{{\rm lc}}
\def\init{{\rm init}}

\def\max{{\rm max}}
\def\s{{\bar{s}}}

\def\sp{{\rm SPoly}}
\def\csp{{\rm CSPoly}}
\def\syz{{\rm Syz}}

\newcommand{\ignore}[1]{}

\begin{document}

\begin{frontmatter}



\title{On Computing \gr Bases in
Rings of Differential Operators}


\author{Xiaodong Ma, Yao Sun and Dingkang Wang}

\address{Key Laboratory of Mathematics Mechanization, Academy of Mathematics and Systems Science, CAS, Beijing 100190,  China\\
(maxiaodong, sunyao)@amss.ac.cn, dwang@mmrc.iss.ac.cn}

\begin{abstract}
Insa and Pauer presented  a basic theory of \gr basis for
differential operators  with coefficients in a commutative ring in
1998, and a criterion was proposed to determine if  a set of
differential operators is a \gr basis. In this paper, we will give
a new criterion such that Insa and Pauer's criterion could be
concluded as a special case and one could compute the \gr basis
more efficiently by this new criterion.
\end{abstract}


\end{frontmatter}



\section{Introduction}

Many investigations have been done on \gr basis in rings of
differential operators \citep{Adams94, Bjork79, Galligo85, Mora86,
Oaku94}, but the coefficients are in fields (of rational
functions), rings of power series, or rings of polynomials over a
field. For example, Mora gave an introduction to commutative and
non-commutative \gr bases, which includes \gr bases for Wely
algebra \citep{Mora94}. As in Insa and Pauer's paper, the rings of
coefficients in this paper are general commutative rings, which is
the main difference from other existing works.

In Insa and Pauer's paper \citep{Insa98}, the results of
Buchberger on \gr basis in polynomial rings have been extended to
the theory of \gr basis for differential operators. A criterion
was presented to determine if a set of differential operators is a
\gr basis, and a basic method for computing the \gr basis was also
proposed. Pauer generalized the theory to a class of rings which
includes rings of differential operators with coefficients in
noetherian rings \citep{Pauer07}.

For computing the \gr basis of a set of differential operators,
instead of computing the generators of the syzygy module generated
by their initials, Insa and pauer's method needs to compute the
generators of  many syzygy modules generated by their leading
coefficients. Thus, Insa and pauer's method leads to many
unnecessary computations. In order to improve the efficiency, Zhou
and Winkler proposed some techniques to reduce the computations on
the syzygies \citep{Zhou07}.

In this paper, a new criterion is proposed for computing \gr basis
in the ring of differential operators with coefficients in a
general commutative ring.

The new criterion bases on the following simple fact: Let $f$, $g$
be two differential operators, then $$fg=gf +h,$$ where $fg$ and
$gf$ have the same degree, but $h$ has less degree than $fg$ or
$gf$. The above equation implies that even though the
multiplication in the rings of differential operators is not
commutative, the products $fg$ and $gf$ still have the same
initial. According to this fact, it suffices to consider the
generators of the syzygy module in a commutative ring which is
deduced from the ring of differential operators. With these
generators, a new criterion is proposed to determine if a set of
differential operators is a \gr basis.  This new result
generalizes the Insa and Pauer's original theorem such that their
theorem can be concluded as a special case of the new theorem.
Furthermore, the results of this paper can extend naturally to the
rings that preserve the same fact.

Then the proposed criterion also leads to an efficient method for
computing Gr\"obner bases in the rings of differential operators.
This new method considers fewer s-polynomials than those in Insa
and Pauer's method as well as Zhou and Winkler's improved version.
So it is not surprising that this new method will have better
efficiency than others.

This paper is organized as follow. Section 2 includes some
preliminaries of the \gr Basis in the rings of differential
operators. The Insa and Pauer's theorem comes in section 3. In
section 4, the new criterion is presented in detail.  And some
algorithmic problems are discussed in section 5. The paper is
concluded in section 6.

\section{\gr Basis in Rings of Differential Operators}

Let $K$ be a field of characteristic zero, $\N$ the set of
non-negative integers, $n\in \N$ a positive integer and
$K[X]:=K[x_1,\cdots,x_n]$ (resp. $K(X):=K(x_1, \cdots, x_n)$) the
ring of polynomials (resp. the field of rational functions) in $n$
variables over $K$. Let $\xi:K(X)\longrightarrow K(X)$ be the
partial derivative by $x_i$ for $1\le i \le n$.

Let $\R$ be a noetherian $K$-subalgebra of $K(X)$ which is stable
by $\xi$ for $1 \le i \le n$, i.e. $\xi(r)\in \R$ for all $r\in
\R$. Important examples for $\R$ are $K[X]$, $K(X)$ and
$K[X]_M:=\{\frac{f}{g}\in K(X)  \mid  f\in K[X], g\in M\}$ where
$M$ is a subset of $K[X]\setminus \{0\}$ closed under
multiplication.

Assume the linear equations over $\R$ can be solved, i.e.
\begin{enumerate}
\item[(1)] for all $g\in \R$ and all finite subsets $F\subset \R$,
it is possible to decide whether $g$ is an element of
${}_{\R}\langle F \rangle$, and if yes, it is available to obtain
a family $(d_f)_{f\in F}$ in $\R$ such that $g=\sum_{f\in F}d_f
f$; \item[(2)] for all finite subsets $F\subset \R$, a finite
system of generators of the $\R$-module
$$\{(s_f)_{f\in F}  \mid  \sum_{f\in F}s_f f=0, s_f \in \R\}$$ can be
computed.
\end{enumerate}

The partial differential operator $D_i$ is defined as the
restriction of $\xi$ to $\R$ for $1\le i \le n$. Let
$\A:=\R[D]=\R[D_1, \cdots, D_n]$ be the $\R$-subalgebra of
$End_k(\R)$ generated by $id_{\R}=1$ and $D_1, \cdots, D_n$. Then
the ring $\A$ is ``a ring of differential operators with
coefficients in $\R$", while the elements of $\A$ are called
``differential operators with coefficients in $\R$" [Insa and
Pauer 1998]. It is well known that $\A$ is a left-neotherian
associative $\R$-algebra, so the ideals in $\A$ always refer to
the left-ideals of $\A$ in this paper.

$\A$ is a non-commutative $K$-algebra with fundamental relations:
$$x_ix_j=x_jx_i \mbox{, } D_iD_j=D_jD_i \mbox{ for } 1\le i, j \le n,$$
and $$D_ir-rD_i=D_i(r), r\in \R.$$ For a simple example, let
$\A=(k[x_1, x_2])[D_1, D_2]$, then $$x_1x_2=x_2x_1\mbox{, }
D_1D_2=D_2D_1\mbox{ and }D_1x_1x_2-x_1x_2D_1=D_1(x_1x_2)=x_2.$$
And for any $f\in \A$, $f$ can be written uniquely as a finite sum
$$f=\sum_{\alpha \in \N^n} r_{\alpha}D^{\alpha},\mbox{ where }
r_{\alpha}\in \R.$$

Let $\prec$ be an admissible order on $\N^n$, i.e. a total order
on $\N^n$ such that $0\in \N^n$ is the smallest element and
$\alpha \prec \beta$ implies $\alpha+\gamma \prec \beta+\gamma$
for all $\alpha, \beta, \gamma \in \N^n$. Then for a differential
operator $0\not= f=\sum_{\alpha \in \N^n} r_{\alpha}D^{\alpha} \in
\A$, the degree, leading coefficient and initial are defined as:
$$\deg(f):=\max_{\prec}\{\alpha \mid r_{\alpha}\not=0\} \in \N^n,$$
$$\lc(f):=r_{\deg(f)},$$
$$\init(f):=\lc(f)D^{\deg(f)}.$$
If $F$ is a subset of $\A$, define:
$$\deg(F):=\{\deg(f) \mid f\in F, f\not=0\},$$
$$\init(F):=\{\init(f) \mid f\in F, f\not=0\}.$$

It is easy to check $\A$ has the following properties. Let $f, g,
h\in \A$:
\begin{enumerate}
\item[] {\bf Associativity:} $$(fg)h=f(gh).$$

\item[] {\bf Distributivity:} $$f(g+h)=fg+fh\mbox{ and }
(f+g)h=fh+gh.$$
\end{enumerate}
There is another property about $\A$ which will be used frequently
in this paper. Let $\init(f)=r_fD^{\alpha_f}$ and
$\init(g)=r_gD^{\alpha_g}$, $r_f, r_g\in \R$, then
$$\deg(fg)=\deg(f)+\deg(g)\mbox{, }\lc(fg)=\lc(f)\lc(g)\mbox{ and }\init(fg)=r_fr_gD^{\alpha_f}D^{\alpha_g}.$$
Therefore, $\A$ also has a {\bf Quasi-Commutativity:}
$$\deg(fg-gf)\prec \deg(fg)=\deg(gf). $$

Then the \gr basis in the rings of differential operators with
coefficients in $\R$ is defined as:

\begin{define}
Let $\J$ be an ideal in $\A$ and $G$ a finite subset of
$\J\setminus\{0\}$. Then $G$ is a \gr basis of $\J$ w.r.t. $\prec$
iff for all $f\in \J$, $$\lc(f)\in {}_{\R}\langle \lc(g) \mid g\in
G, \deg(f)\in \deg(g)+\N^n \rangle.$$
\end{define}

\begin{example}
If $\J={}_{\A}\langle f \rangle\subset \A$ and $f\not=0$, then
$\{f\}$ is a \gr basis of $\J$.
\end{example}

\section{Insa and Pauer's Theorem}

In order to compute the \gr basis, a division (or reduction) in
$\A$ is necessary. In theory, there may exist various kinds of
divisions in $\A$. The following division is the one presented by
Insa and Pauer in \citep{Insa98}.

\begin{prop}[Division in $\A$] Let $F$ be a finite subset of
$\A\setminus\{0\}$ and $g \in \A$. Then there exist a differential
operator $r\in \A$ and a family $(h_f)_{f\in F}$ in $\A$ such
that:
\begin{enumerate}
\item $g=\sum_{f\in F}h_ff+r, \ \ \mbox{(} r \mbox{  is ``a
remainder of } g \mbox{ after division by } F \mbox{"),}$

\item for all $f\in F$, $h_f=0$ or $\deg(h_ff)\preceq \deg(g), $

\item $r=0$ or $\lc(r) \notin {}_{\R} \langle \lc(f)\mid
\deg(r)\in \deg(f)+ \N^n \rangle.$
\end{enumerate}
\end{prop}

This definition of division in $\A$ is also used in the new
theorem presented in the next section. Based on this division, a
\gr basis in $\A$ has the following property \citep{Insa98}.

\begin{prop} \label{Insaporp2}
Let $\J$ be an ideal in $\A$, $G$ a \gr basis of $\J$ and $f\in
\A$. Then $f\in \J$ iff a remainder of $f$ after division by $G$
is zero.
\end{prop}

Then the next theorem proposed by Insa and Pauer provides a
criterion for checking if a set of differential operators is a \gr
basis.

\begin{theorem}[Insa and Pauer's theorem] \label{theoremIP}
Let $G$ be a finite subset of $\A\setminus \{0\}$ and $\J$ the
ideal in $\A$ generated by $G$. For $E\subset G$, let $S_E$ be a
finite set of generators of the $\R$-module
$$\syz_{\R}(E):=\{(s_e)_{e\in E}  \mid  \sum_{e\in E}s_e
\lc(e)=0\}\subset {}_{\R}(\R^{|E|}).$$ Then the following
assertions are equivalent:
\begin{enumerate}
\item $G$ is a \gr basis of $\J$.

\item For all $E\subset G$ and for all $(s_e)_{e\in E}\in S_E$, a
remainder of $$\sp(E, (s_e)_{e\in E}):=\sum_{e\in
E}s_eD^{m(E)-\deg(e)}e$$ after division by $G$ is zero, where
$$m(E):=(\max_{e\in E}\deg(e)_1, \cdots, \max_{e\in
E}\deg(e)_n)\in \N^n.$$
\end{enumerate}
\end{theorem}

According to this theorem, one is able to compute the \gr basis of
${}_\A \langle F\rangle$ for any subset $F\subset \A$. All needed
to do is to check the remainder of $\sum_{e\in
E}s_eD^{m(E)-\deg(e)}e$ after division by $F$ is zero or not for
all $E\subset F$. If there does exist a remainder $r$ which is not
zero, then expand $F$ to $F':=F\cup \{r\}$ and repeat the process
for $F'$. The procedure terminates exactly when all the remainders
are zero. The terminality of this algorithm can be proved in a
similar way as the general \gr basis algorithm.

During the above computing process, in order to seek non-zero
remainders w.r.t. the subsets of $F$, one needs to compute the
generators of $Syz(E)$ for all $E\subset F$, which is really
expensive. In view of this, Zhou and Winkler proposed a trick to
avoid some unnecessary computations \citep{Zhou07}. In their
paper, they show that if the elements in $E$ have some special
properties, then instead of computing the generators of $Syz(E)$,
it only suffices to calculate the generators of $Syz(E')$ for some
$E'\subset E$. Since the new theorem in the current paper
generalizes Insa and Pauer's theorem in a different way from Zhou
and Winkler, the details of their method are omitted here. For
interesting readers, please see \citep{Zhou07}.

\section{The New Theorem for \gr Basis in Rings of Differential Operators}

The differential operator $$\sp(E, (s_e)_{e\in E})=\sum_{e\in
E}s_eD^{m(E)-\deg(e)}e$$ in (ii) of the Insa and Pauer's theorem
is denoted as a ``generalized s-polynomial" w.r.t. the subset
$E\subset G$ in \citep{Zhou07}, as it plays the same role as the
general s-polynomials.

However, this generalized s-polynomial in Insa and Pauer's theorem
is constructed quite strangely, since it is not created by the
syzygies of $\init(G)$ in the traditional way but results from the
set $S_E$, which is a set of generators of $\{(s_e)_{e\in E} \mid
\sum_{e\in E}s_e \lc(e)=0\}$. With a further study, one will find
the reason easily. That is, the syzygy of $\init(G)$ is extremely
difficult to define and even harder to compute, as $\A$ is a
non-commutative ring. This explains why Insa and Pauer concentrate
on the syzygy of $\lc(E)$ in $\R$ instead.

At this point, it is natural to ask: {\bf do we really need the
syzygy of ${\bf init}(G)$?} The answer is {\bf NO!}

By revisiting the proof of Insa and Pauer's theorem carefully, in
order to show $G$ is a \gr basis, it suffices to consider the
differential operators which are generated by $G$ and possibly
have new initials. So all we need to do is to eliminate the
present initials of $G$ and to try to create all possible new
initials in ${}_{\A}\langle G \rangle$ . Fortunately, the syzygy
of $\init(G)$ is not the only one that could do this job, since
the ring $\A$ has the Quasi-Commutativity.

With these considerations in mind, let discuss a commutative ring
$\B$ first which is deduced from the Quasi-Commutative ring $\A$.

Let $\B:=\R[Y]=\R[y_1, \cdots, y_n]$, which is generated by
$id_{\R}=1$ and $y_1, \cdots, y_n$. $\B$ is a commutative
$K$-algebra with fundamental relations: $$x_ix_j=x_jx_i \mbox{, }
y_iy_j=y_jy_i \mbox{ and } x_iy_j=y_jx_i \mbox{ for } 1\le i, j
\le n.$$ For any $f\in \B$, $f$ can also be written uniquely as a
finite sum $f=\sum_{\alpha \in \N^n} r_{\alpha}Y^{\alpha}$, where
$r_{\alpha}\in \R$. Similarly, the degree, leading coefficient and
initial are defined as: $\deg(f):=\max_{\prec}\{\alpha \mid
r_{\alpha}\not=0\} \in \N^n$, $\lc(f):=r_{\deg(f)}$ and
$\init(f):=\lc(f)Y^{\deg(f)}$ respectively.

Since $Y$ commute with $X$ and the linear equations over $\R$ are
solvable, it is easy to check the linear equations over $\B$ can
be solved as well, which means the generators of
$$\syz_{\B}(F):=\{(s_f)_{f\in F}  \mid  \sum_{f\in F} s_f \init(f)=0, s_f\in
\B\}$$ can be computed, where $F\subset \B\setminus \{0\}$.

With a little care, the only difference between $\B$ and $\A$ is
that $\B$ is commutative and $\A$ is not. The following map
bridges the two rings easily. Let $\sigma$ be a map from $\B$ to
$\A$ such that for any $\sum_{\alpha\in
\N^n}r_{\alpha}Y^{\alpha}\in \B$ where $r_{\alpha}\in \R$,
$$\sigma(\sum_{\alpha\in \N^n}r_{\alpha}Y^{\alpha}) =
\sum_{\alpha\in \N^n}r_{\alpha}D^{\alpha}\in \A.$$ By the
definition of $\sigma$, the following properties hold for all $f,
g\in \B$: $$\deg(f)=\deg(\sigma(f)),$$
$$\lc(f)=\lc(\sigma(f)),$$
$$\sigma(\init(f))=\init(\sigma(f)),$$
$$\deg(fg)=\deg(\sigma(fg))=\deg(\sigma(f)\sigma(g)),$$
$$\lc(fg)=\lc(\sigma(fg))=\lc(\sigma(f)\sigma(g)),$$
$$\sigma(\init(fg))=\init(\sigma(fg))=\init(\sigma(f)\sigma(g)).$$
But remark that $$\sigma(fg)\not=\sigma(f)\sigma(g).$$ It is also
very easy to check $\sigma$ is a {\bf $\R$-homomorphism}, i.e. for
$f, g\in \B$ and $r\in \R$,
$$\sigma(rf+g)=r\sigma(f)+\sigma(g).$$ All the above properties
will be used frequently in the proof of the new theorem.

Before presenting the new theorem, let study some properties of
the ring $\B$ first. These properties will be used in the proof of
the new theorem as well. We start with the following definition.

\begin{define}
An element $(s_f)_{f\in F}\in S(F)$ is {\bf homogeneous of degree
$\alpha$}, where $\alpha \in \N^n$, provided that $$(s_f)_{f\in
F}=(c_fY^{\alpha_f})_{f\in F},$$ where $c_f\in \R$ and
$\alpha_f+\deg(f)=\alpha$ whenever $c_f\not=0$.
\end{define}

The following two lemmas are well-known. For details, please see
\citep{cox}.

\ignore{
\begin{lemma}
Every element of $S(F)$ can be written uniquely as a sum of
homogeneous elements of $S(F)$.
\end{lemma}}

\begin{lemma}
$\syz_{\B}(F)$ has a set of homogeneous generators, i.e. there
exists a finite set $C_F\subset S(F)$ such that each element of
$C_F$ is homogeneous and $\syz_{\B}(F)= {}_{\B}\langle C_F
\rangle$.
\end{lemma}

\begin{lemma}\label{homo}
Let $C_F$ be a set of homogeneous generators of $\syz_{\B}(F)$. If
$(s_f)_{f\in F}\in \syz_{\B}(F)$ is homogeneous of degree
$\alpha$, then there exists a family $(r_\s)_{\s\in C_F}$ where
$r_\s\in \B$, such that
$$(s_f)_{f\in F}=\sum_{\s\in C_F}r_\s \s$$ and $r_\s \s$ is
homogeneous of degree $\alpha$ for all $\s\in C_F$.
\end{lemma}

Now, it is time to present the new theorem.

\begin{theorem}[Main theorem]\label{mainthm}
Let $G$ be a finite subset of $\A\setminus \{0\}$ and $\J$ the
ideal in $\A$ generated by $G$. For each $g\in G$, assume
$\init(g)=c_gD^{\alpha_g}$ where $c_g\in \R$ and $\alpha_g\in
\N^n$. Let $C_G$ be a set of homogeneous generators of
$\syz_{\B}(H_G)$ where $H_G=\{c_gY^{\alpha_g} \mid g\in G\}\subset
\B$ and $C_G$ is called a set of commutative syzygy generators of
$\init(G)$ for short. Then the following assertions are
equivalent:
\begin{enumerate}
\item $G$ is a \gr basis of $\J$.

\item For all $(s_g)_{g \in G}\in C_G$ where $s_g\in \B$ and hence
$\sigma(s_g)\in \A$, a remainder of
$$\csp((s_g)_{g \in G}):=\sum_{g\in G}\sigma(s_g)g$$ after division by $G$ is
zero.
\end{enumerate}
\end{theorem}

\begin{proof}
(i)$\Rightarrow$(ii): It follows from Proposition \ref{Insaporp2}.

(ii)$\Rightarrow$(i): Let $h\in \J$. It suffices to show:
$$\lc(h)\in {}_{\R}\langle \lc(g)\mid g\in G, \deg(h)\in deg(g)+\N^n \rangle.$$
For a family $(f_g)_{g\in G}$ in $\A$, define
$$\delta((f_g)_{g\in G}):=\max_{\prec}\{\deg(f_g)+\deg(g)\mid g\in G\}.$$

Since $h\in \J$, there exists a family $(h_g)_{g\in G}$ in $\A$
such that $h=\sum_{g\in G}h_gg$. Choose $(h_g)_{g\in G}$ such that
$$\delta := \delta((h_g)_{g\in G}) \mbox{ is minimal,}$$ which
implies if $(h'_g)_{g\in G}$ is such that $h=\sum_{g\in G}h'_gg$,
then $\delta\preceq \delta((h'_g)_{g\in G})$.

Let $E:=\{g\in G\mid \deg(h_g)+\deg(g)=\delta\}\subset G$.

\noindent {\bf  Case 1}: $\deg(h)=\delta$. Then
$$\init(h)=\sum_{g\in E} \init(h_gg) \mbox{ and }
\lc(h)=\sum_{g\in E} \lc(h_g)\lc(g)\in {}_{\R}\langle \lc(g) \mid
g\in E \rangle.$$ If $g\in E$, then
$\deg(h)=\delta=\deg(h_g)+\deg(g)$ and hence $\deg(h)\in
\deg(g)+\N^n$. Therefore, $\lc(h)\in {}_{\R}\langle \lc(g)\mid
g\in G, \deg(h)\in \deg(g)+\N^n \rangle$.

\noindent{\bf  Case 2}: $\deg(h)\prec \delta$. Then
$$\sum_{g\in E} \init(h_gg)=0, \mbox{ which implies } \sum_{g\in E} \lc(h_g)\lc(g)=0.$$
Combined with the fact that $\deg(h_g)+\deg(g)=\delta$ for $g\in
E$, it follows $$0=\sum_{g\in E}\lc(h_g)\lc(g)Y^\delta=\sum_{g\in
E}\lc(h_g)Y^{\deg(h_g)}\lc(g)Y^{\deg(g)}\in \B.$$ Denote
$$t_g:=\left\{\begin{array}{cl} \lc(h_g)Y^{\deg(h_g)},  & \ \ \  g\in E, \\ 0, & \ \ \ g\in G\setminus E.
\end{array}\right. $$
Notice that
$$\sigma(t_g):=\left\{\begin{array}{cl} \init(h_g),  & \ \ \  g\in E, \\ 0, & \ \ \ g\in G\setminus E.
\end{array}\right. $$
Then $(t_g)_{g\in G}$ is a homogeneous element of $\syz_{\B}(H_G)$
with degree $\delta$. Since $C_G$ is a set of homogeneous
generators of $\syz_{\B}(H_G)$, by lemma \ref{homo}, there exists
a family $(r_\s)_{\s\in C_G}$ where $r_\s\in \B$, such that
$(t_g)_{g\in G}=\sum_{\s\in C_G}r_\s \s$ and $r_\s \s$ is
homogeneous of degree $\delta$, i.e. for $\forall g\in G$,
$$t_g=\sum_{\s\in C_G}r_\s s_g, \mbox{ where } \s=(s_g)_{g\in G},$$
and for $\forall g\in G$, $\forall \s\in C_G$,
$$\deg(r_\s)+\deg(s_g)+\deg(g)=\delta \mbox{ whenever } r_\s s_g\not=0.$$
Remark that all $t_g, r_\s, s_g\in \B$.

Now $$h=\sum_{g\in G}h_g g=\sum_{g\in E}h_g g+\sum_{g\in
G\setminus E}h_g g$$ \begin{equation}\label{eq1} =(\sum_{g\in
E}h_g g - \sum_{g\in G} \sum_{\s\in C_G} \sigma(r_\s) \sigma(s_g)
g) + \sum_{g\in G} \sum_{\s\in C_G} \sigma(r_\s) \sigma(s_g) g +
\sum_{g\in G\setminus E}h_g g.\end{equation}

For the {\bf FIRST} sum in (\ref{eq1}), $$\sum_{g\in E}h_g g -
\sum_{g\in G} \sum_{\s\in C_G} \sigma(r_\s) \sigma(s_g) g =
\sum_{g\in E}\init(h_g) g - \sum_{g\in G} \sum_{\s\in C_G}
\sigma(r_\s) \sigma(s_g) g + \sum_{g\in E}(h_g - \init(h_g)) g$$
$$=\sum_{g\in G}\sigma(t_g) g - \sum_{g\in G} \sum_{\s\in C_G}
\sigma(r_\s) \sigma(s_g) g + \sum_{g\in E}(h_g - \init(h_g)) g$$
$$=\sum_{g\in G} (\sigma(t_g)  - \sum_{\s\in C_G} \sigma(r_\s)
\sigma(s_g)) g + \sum_{g\in E}(h_g - \init(h_g)) g$$
\begin{equation}\label{eq2}
=\sum_{g\in G} \sum_{\s\in C_G} (\sigma(r_\s s_g) - \sigma(r_\s)
\sigma(s_g)) g + \sum_{g\in E}(h_g - \init(h_g)) g.\end{equation}
Since $\init(\sigma(r_\s s_g)) = \init(\sigma(r_\s) \sigma(s_g))$,
then for $\forall g\in G$, $\forall \s\in C_G$,
$$\deg((\sigma(r_\s s_g) - \sigma(r_\s)\sigma(s_g))g) \prec
\deg(\sigma(r_\s)\sigma(s_g) g) = \deg(r_\s)+\deg(s_g)+\deg(g)
=\delta,$$ whenever $r_\s s_g\not=0$. In case of $r_\s s_g=0$ and
$\sigma(r_\s)\sigma(s_g)\not=0$, $\lc(r_\s s_g)=0$ implies
$\lc(\sigma(r_\s)\sigma(s_g))=0$, so the above inequation holds as
well. Besides, clearly for $\forall g\in E$,
$$\deg((h_g - \init(h_g)) g)\prec \deg(h_gg)=\delta.$$

For the {\bf SECOND} sum in (1), $$\sum_{g\in G} \sum_{\s\in C_G}
\sigma(r_\s) \sigma(s_g) g = \sum_{\s\in C_G} \sigma(r_\s)
(\sum_{g\in G} \sigma(s_g) g).$$ For each $\s=(s_g)_{g\in G}\in
C_G$, assume $\s$ is homogeneous of degree ${\beta_{\s}}$, then
${\beta_{\s}}=\deg(\sigma(s_g))+\deg(g)$ whenever
$\sigma(s_g)\not=0$, and consider
$$\sum_{g\in G} \sigma(s_g) g=\sum_{g\in G}
\init(\sigma(s_g)g)+\sum_{g\in G} (\sigma(s_g)g-
\init(\sigma(s_g)g))$$ $$= \sum_{g\in G}
\lc(\sigma(s_g))\lc(g)D^{\beta_{\s}}+\sum_{g\in G} (\sigma(s_g)g-
\init(\sigma(s_g)g)).$$ By the definition of $C_G$ and $\s$ is a
homogeneous element of $\syz_{\B}(H_G)$ with degree
${\beta_{\s}}$, then
$$0=\sum_{g\in G}s_gc_gY^{\alpha_g}=\sum_{g\in
G}\lc(s_g)c_gY^{\beta_{\s}} \mbox{ where
}\init(g)=c_gD^{\alpha_g}.$$ Notice that
$\lc(\sigma(s_g))=\lc(s_g)$, which implies
$$\sum_{g\in G} \lc(\sigma(s_g))\lc(g)D^{\beta_{\s}}=0.$$ Combined with
the fact $\deg(\sigma(s_g)g- \init(\sigma(s_g)g))\prec
{\beta_{\s}}$, the following inequation holds: $$\deg(\sum_{g\in
G} \sigma(s_g) g)\prec {\beta_{\s}}.$$ By (ii) a remainder of
$\sum_{g\in G} \sigma(s_g) g$ after division by $G$ is zero, i.e.
there exist families $(f_g(\s))_{g\in G}$ in $\A$, such that
$$\sum_{g\in G} \sigma(s_g) g = \sum_{g\in G}f_g(\bar{s})g,$$ and
$\deg(f_g(\bar{s}) g)\preceq \deg(\sum_{g\in G} \sigma(s_g)
g)\prec {\beta_{\s}}$. So the second sum in (1) turns out to be
 $$\sum_{g\in G} \sum_{\s\in C_G}
\sigma(r_\s) \sigma(s_g) g = \sum_{\s\in C_G} \sigma(r_\s)
(\sum_{g\in G} \sigma(s_g) g) = \sum_{\s\in C_G} \sigma(r_\s)
(\sum_{g\in G}f_g(\s)g) $$ \begin{equation} \label{eq3} =
\sum_{g\in G} \sum_{\s\in C_G} \sigma(r_\s)f_g(\s)g\end{equation}
and for $\forall g\in G$, $\forall \s\in C_G$,
$$\deg(\sigma(r_\s)f_g(\s)g)\prec \deg(\sigma(r_\s))+\beta_\s =
\delta \mbox{ whenever } r_\s \not=0.$$

For the {\bf THIRD} sum in (1), by the definition of $E$, it is
obvious that $\deg(h_gg)\prec \delta$ for $g\in G\setminus E$.

Based on the expressions in (2) and (3), let
$$h'_g:=\left\{\begin{array}{cl} \sum_{\s\in C_G} (\sigma(r_\s s_g) - \sigma(r_\s)
\sigma(s_g)+\sigma(r_\s)f_g(\s))+(h_g - \init(h_g)), & \ \ \  g\in
E,
\\ \sum_{\s\in C_G} (\sigma(r_\s s_g) - \sigma(r_\s)
\sigma(s_g)+\sigma(r_\s)f_g(\s))+h_g, & \ \ \ g\in G\setminus E.
\end{array}\right. $$
Then it is easy to verify that $h=\sum_{g\in G} h'_gg$ and
$\delta((h'_g)_{g\in G})\prec \delta$, which is a contradiction to
the minimality of $\delta$. Hence case 2 never occurs.

To sum up, the theorem is proved.
\end{proof}

The above theorem provides a more fundamental criterion than Insa
and Pauer's original theorem, since it suffices to consider the
``s-polynomials" constructed from a set of commutative syzygy
generators of $\init(G)$. As we will see in the next section, Insa
and Pauer's original theorem only provides a method for computing
the set $C_G$. Thus the new theorem is more essential and the Insa
and Pauer's theorem can be concluded as its natural corollary.

In fact, the main theorem extends much more generally.

\begin{theorem}
The main theorem is true for all rings with the quasi-commutative
property.
\end{theorem}

\begin{proof}
In the proof of the main theorem, only the quasi-commutative
property is used.
\end{proof}

Similar to the Insa and Pauer's approach, one can also develop an
algorithm for computing \gr basis of ${}_{\A} \langle F \rangle$
for any given $F\subset \A$ based on the main theorem.  According
to theorem \ref{mainthm}, it suffices to compute {\bf one} set of
commutative syzygy generators of $\init(F)$ in the commutative
ring $\B$, instead of computing the generators of $\syz_{\R}(E)$
for {\bf all} subsets $E\subset F$. Clearly, Insa and Pauer's
method leads to more computations than needed. To illustrate this,
let see the following example which is from \citep{Zhou07}.

\begin{example} \label{exam}
Let $\R=\Q[x_1,\cdots, x_6]$, $\A=\R[D_1, \cdots, D_6]$ and J the
left ideal of $\A$ generated by $F=\{f_1, f_2, f_3, f_4\}$, where
$f_1=x_1D_4+1, f_2=x_2D_5, f_3=(x_1+x_2)D_6, f_4=D_5D_6$. Let
$\prec$ be the graded lexicographic order with $(1, 0, \cdots, 0)
\prec (0, 1, \cdots, 0) \prec (0, \cdots, 0, 1)$.

By Insa and Pauer's theorem, in order to compute a \gr basis for
${}_{\A} \langle F \rangle$, one needs to consider the following
``generalized s-polynomials" (duplicated cases are omitted):
$$\sp(\{f_1, f_2\}, (x_2, -x_1))=x_2D_5f_1-x_1D_4f_2,$$
$$\sp(\{f_1, f_3\}, (x_1+x_2, -x_1))=(x_1+x_2)D_6f_1-x_1D_4f_3,$$
$$\sp(\{f_1, f_4\}, (1, -x_1))=D_5D_6f_1-x_1D_4f_4,$$
$$\sp(\{f_2, f_3\}, (x_1+x_2, -x_2))=(x_1+x_2)D_6f_2-x_2D_5f_3,$$
$$\sp(\{f_2, f_4\}, (1, -x_2))=D_6f_2-x_2f_4,$$
$$\sp(\{f_3, f_4\}, (1, -(x_1+x_2)))=D_5f_3-(x_1+x_2)f_4,$$
$$\sp(\{f_1, f_2, f_3\}, (x_2, -x_1, 0))=x_2D_5D_6f_1-x_1D_4D_6f_2,$$
$$\sp(\{f_1, f_2, f_3\}, (1, 1, -1))=D_5D_6f_1+D_4D_6f_2-D_4D_5f_3,$$
$$\sp(\{f_1, f_2, f_4\}, (0, 1, -x_2))=D_4D_6f_2-x_2D_4f_4,$$
$$\sp(\{f_1, f_3, f_4\}, (x_1+x_2, -x_1, 0))=(x_1+x_2)D_5D_6f_1-x_1D_4D_5f_4,$$
$$\sp(\{f_1, f_3, f_4\}, (1, -1, x_2))=D_5D_6f_1-D_4D_5f_3+x_2D_4f_4,$$
$$\sp(\{f_2, f_3, f_4\}, (1, -1, x_1))=D_6f_2-D_5f_3+x_1f_4.$$

By Zhou and Winkler's trick, $\sp(\{f_1, f_2, f_4\}, (0, 1,
-x_2))$, $\sp(\{f_1, f_3, f_4\}, (x_1+x_2, -x_1, 0))$, $\sp(\{f_1,
f_3, f_4\}, (1, -1, x_2))$ and $\sp(\{f_2, f_3, f_4\}, (1, -1,
x_1))$ can be removed.

However, according to the new theorem, $\B=\R[y_1, \cdots, y_6]$
and $H_F=\{x_1y_4, x_2y_5, (x_1+x_2)y_6, y_5y_6\}$. Then
$$C_F=\{\s_1, \s_2, \s_3, \s_4, \s_5\}=\{(x_2y_5, -x_1y_4, 0, 0), ((x_1+x_2)y_6, 0, -x_1y_4, 0),$$
$$(y_5y_6, 0, 0, -x_1y_4), (0, y_6, 0, -x_2), (0, 0, y_5,
-(x_1+x_2))\},$$ is a set of commutative syzygy generators of
$\init(F)$. Therefore, in the new method, it suffices to consider:
$$\csp(\s_1)=x_2D_5f_1-x_1D_4f_2,$$
$$\csp(\s_2)=(x_1+x_2)D_6f_1-x_1D_4f_3,$$
$$\csp(\s_3)=D_5D_6f_1-x_1D_4f_4,$$
$$\csp(\s_4)=D_6f_2-x_2f_4,$$
$$\csp(\s_5)=D_5f_3-(x_1+x_2)f_4.$$

\ignore{
$$\csp(\s_1)=x_2D_5f_1-x_1D_4f_2=\sp(\{f_1, f_2\}, (x_2, -x_1)),$$
$$\csp(\s_2)=(x_1+x_2)D_6f_1-x_1D_4f_3=\sp(\{f_1, f_3\}, (x_1+x_2, -x_1)),$$
$$\csp(\s_3)=D_5D_6f_1-x_1D_4f_4=\sp(\{f_1, f_4\}, (1, -x_1)),$$
$$\csp(\s_4)=D_6f_2-x_2f_4=\sp(\{f_2, f_4\}, (1, -x_2)),$$
$$\csp(\s_5)=D_5f_3-(x_1+x_2)f_4=\sp(\{f_3, f_4\}, (1,
-(x_1+x_2))).$$}

No matter in either Insa and Pauer's method or Zhou and Winkler's
improved version, one has to compute the remainders of $\sp(\{f_2,
f_3\}, (x_1+x_2, -x_2))$ and $\sp(\{f_1, f_2, f_3\},\\ (1, 1,
-1))$ all the time, which are not needed any more in the new
method. Therefore, the new method avoids all these unnecessary
computations and hence has better efficiency.

To finish this example, it is easy to check that all the
remainders of $\csp(\s_i)$ after division by $F$ are zero. So $F$
itself is a \gr basis for ${}_{\A}\langle F \rangle$.
\end{example}

\section{On Computing $C_G$ over $\R[Y]$}

So far, as shown by the main theorem \ref{mainthm}, in order to
check if a set of differential operators $G$ is a \gr basis for
${}_{\A}\langle G \rangle$, it only needs to consider the
``s-polynomials" deduced by $C_G$, which is a set of commutative
syzygy generators of $\init(G)$. Now the last question is {\bf how
to compute the set $C_G$ over $\R[Y]$?}

By the definition of $C_G$, it is a set of homogeneous generators
of $\syz_{\B}(H_G)$ which is a syzygy module of monomials in
$\B=\R[Y]$. In fact, Insa and Pauer's theorem implies a natural
method to compute it. That is, the set
$$\{(s_eY^{m(E)-\deg(e)})_{e\in E} \mid (s_e)_{e\in E}\in S_E, E\subset G \},$$
where $S_E$ is a set of generators of $\syz_{\R}(E)=\{(s_e)_{e\in
E} \mid \sum_{e\in E}s_e \lc(e)=0, s_e\in \R\}$ and
$m(E)=(\max_{e\in E}\deg(e)_1, \cdots, \max_{e\in E}\deg(e)_n)\in
\N^n,$ extends to a set of generators of $\syz_{\B}(H_G)$
naturally. But example \ref{exam} shows this set is not minimal in
general.

Since $\B=\R[Y]$ is a commutative ring, there are many
sophisticated results on computing the syzygy of monomials in
$\B$, such as the techniques in \citep{Adams94}. Also Zhou and
Winkler's trick can be exploited for this purpose. Here we only
mention two special cases.

\begin{enumerate}
\item $\R$ is a field:

When $\R$ is a field, the following set $$\{(\lc(g)Y^{m(f,
g)-\deg(f)}, -\lc(f)Y^{m(f, g)-\deg(g)}) \mid f, g\in G\}$$
extends to a set of generators of $\syz_{\B}(H_G)$.

\item $\R$ is the polynomial ring $K[X]$:

Since the variables $X$ commute with $Y$, $C_G$ can be obtained by
computing the generators of $\syz_{\B}(H_G)$ in the polynomial
ring $K[X, Y]$. Notice that $H_G=\{c_gY^{\alpha_g} \mid g\in G\}
\subset K[X,Y]$. We can obtain a finite set of generators for
$\{(s_g)_{g\in G} \mid \sum_{g\in G} s_g c_gY^{\alpha_g}=0, s_g\in
K[X, Y]\}$ in the polynomial ring $K[X, Y]$ and denote it by $S$.
It is straightforward to check that $S$ is also a set of
generators for $\syz_{\B}(H_G)$ when considered in $K[X][Y]$. Then
the collection of all homogeneous parts of $S$ is a set of
homogeneous generators for $\syz_{\B}(H_G)$, since
$\syz_{\B}(H_G)$ itself is a graded syzygy module in
$(K[X][Y])^{|H_G|}$.
\end{enumerate}

\section{Conclusion}

In this paper, a new theorem which determines if a set of
differential operators is a \gr basis in the ring of differential
operators is proposed. This new theorem is so essential that the
original Insa and Pauer's theorem can be concluded as its natural
corollary. Furthermore, based on the new theorem, a new method for
computing \gr basis in rings of differential operators is deduced.
The new method avoids many unnecessary computations naturally and
hence has better efficiency than other well-known methods.


\begin{thebibliography}{20}



\bibitem[Adams and Loustaunau, 1994]{Adams94}
Adams, W., Loustaunau, P., 1994. An Introduction to \gr Bases.
American Mathematical Society, Providence.

\bibitem[Bj\"ork, 1979]{Bjork79}
Bj\"ork, J., 1979.
Rings of differential operators. North-Holland Pub. Comp.,
Amsterdam, Oxforde, New York.

\bibitem[Cox et al., 1996]{cox}
Cox, D., Little, J., O'Shea, D., 1996. Ideals, Varieties and
Algorithms. Second Edition, Springer, ISBN 0-387-94680-2.

\bibitem[Galligo, 1985]{Galligo85}
Galligo, A., 1985. Some algorithmic questions on ideals of
differential operators. In: Springer Lecture Notes in Computer
Science, vol. 204. pp. 413-421.

\bibitem[Insa and Pauer, 1998]{Insa98}
Insa, M., Pauer, F., 1998. \gr bases in rings of differential
operators. In: Buchberger, B.,Winkler, F. (Eds.), \gr Bases and
Applications. Cambridge University Press, Cambridge.

\bibitem[Mora, 1986]{Mora86}
Mora, F., 1986.
\gr bases for non-commutative polynomial rings. In J. Calmet
(ed.), Proc. AAECC-3,LNCS (229), pp. 353-362.

\bibitem[Mora, 1994]{Mora94}
Mora, T., 1994.
An introduction to commutative and noncommutative \gr bases.
Theoret. Comput. Sci. 134(1), pp. 131-173.

\bibitem[Oaku and Shimoyama, 1994]{Oaku94}
Oaku, T., Shimoyama, T., 1994. A \gr basis method for modules over
rings of differential operators. J. Symbolic Computation (18/3),
pp. 223-248.

\bibitem[Pauer, 2007]{Pauer07}
Pauer, F., 2007.
\gr bases with coefficients in rings. J. Symbolic Computation
(42/11-12), pp. 1003-1011.

\bibitem[Zhou and Winkler, 2007]{Zhou07}
Zhou, M., Winkler, F., 2007. On Computing \gr Bases in Rings of
Differential Operators with Coefficients in a Ring. Mathematics in
Computer Science, col 1, pp. 211-223.

\end{thebibliography}
\end{document}